\documentclass[10pt,conference,final]{IEEEtran}

\newcommand{\e}{\mathrm{e}}

\usepackage{cite}
\usepackage{url}
\usepackage{amsmath}
\usepackage{amssymb}
\usepackage{amsthm}
\usepackage{graphicx}

\newtheorem*{mainresult}{Main Result}
\newtheorem{lemma}{Lemma}
\theoremstyle{definition}

\theoremstyle{remark}
\newtheorem{remark}{Remark}

\newcommand{\E}{\mathsf{E}}

\newcommand{\R}{\mathbb{R}}
\newcommand{\vm}[1]{\boldsymbol{#1}}
\newcommand{\trans}{\mathsf{T}}
\newcommand{\im}{\mathrm{i}}

\newcommand{\M}{M}
\newcommand{\N}{N}
\newcommand{\K}{K}

\newcommand{\NR}{u}
\newcommand{\nr}{a}
\newcommand{\nrb}{b}
\newcommand{\nrbra}{[\nr]}
\newcommand{\nrbbra}{[\nrb]}

\newcommand{\Qhat}{\hat{Q}}
\newcommand{\mhat}{\hat{m}}
\newcommand{\chihat}{\hat{\chi}}

\begin{document}

\title{Analysis of Sparse Representations Using
  Bi-Orthogonal Dictionaries}

\author{\IEEEauthorblockN{Mikko Vehkaper{\"a}$^{1,3}$, Yoshiyuki Kabashima$^{2}$,
Saikat Chatterjee$^{1}$, Erik Aurell$^{1,3}$, \\ Mikael Skoglund$^{1}$ and Lars Rasmussen$^{1}$} 
\vspace*{0ex}
\IEEEauthorblockA{$^{1}$%
KTH Royal Institute of Technology and the ACCESS Linnaeus Center, 
SE-100 44, Stockholm, Sweden}
\IEEEauthorblockA{$^{2}$%
Tokyo Institute of Technology, Yokohama, 226-8502, Japan
}%
\IEEEauthorblockA{$^{3}$%
Aalto University, P.O. Box 11000, FI-00076 AALTO, Finland
}%
E-mails: \url{{mikkov, sach, eaurell, skoglund, lkra}@kth.se},
\url{kaba@dis.titech.ac.jp}
\vspace*{-0.5cm}
}

\maketitle

\begin{abstract}
The sparse representation problem of
recovering an $\N$ dimensional sparse vector $\vm{x}$
from $\M < \N$ linear observations $\vm{y} = \vm{D} \vm{x}$ given
dictionary $\vm{D}$ is considered.  
The standard approach is to let the elements of the dictionary
be independent and identically distributed (IID)
zero-mean Gaussian and minimize the 
$l_{1}$-norm of $\vm{x}$ under the constraint $\vm{y} = \vm{D} \vm{x}$.
In this paper, the  performance of $l_{1}$-reconstruction
is analyzed, when the dictionary is
bi-orthogonal $\vm{D} = [\vm{O}_{1} \; \vm{O}_{2}]$, where
$\vm{O}_{1}, \vm{O}_{2}$
are independent and drawn uniformly according to the 
Haar measure on the group of orthogonal $\M\times\M$ matrices.
By an application of the replica method, we obtain the 
critical conditions under which perfect $l_{1}$-recovery
is possible with bi-orthogonal dictionaries.
\end{abstract}

\section{Introduction}

The sparse representation (SR)
problem has wide applicability,
for example, in communications
\cite{Fletcher_ISIT_2009, 
Gastpar_Sastry_2010_Distributed_sensor_perception},
multimedia \cite{Elad_book_2010}, 
and compressive sampling (CS)
\cite{Donoho_2006_Compressed_sensing,
CS_introduction_Candes_Wakin_2008}.
The standard SR problem is to find 
the sparsest $\vm{x} \in \mathbb{R}^{\N}$ that is the solution to
the set of $\M < \N$ linear equations
\begin{equation}
\vm{y} = \vm{D} \vm{x},
\label{eq:Sparse_Representation_without_Noise}
\end{equation}
for a given dictionary or sensing matrix $\vm{D} \in \mathbb{R}^{\M \times \N}$ 
and observation $\vm{y}$.  Finding such $\vm{x}$ is, however, 
non-polynomial (NP) hard.  Thus, a variety of practical
algorithms have been developed that solve the SR problem sub-optimally. 
The topic of the current paper is the convex relaxation 
approach where, instead of searching for the $\vm{x}$ having
the minimum $l_{0}$-norm,
the goal is to find the minimum $l_{1}$-norm solution of
\eqref{eq:Sparse_Representation_without_Noise}.

Let $\K$ be the number of non-zero elements in $\vm{x}$ and
assume that the convex relaxation method is used for recovery.
The trade-off between two parameters  
$\rho = \K / \N$ and $\alpha = \M / \N$ is then of special interest 
since it tells how much the sparse signal can 
be compressed under $l_{1}$-reconstruction.  
An interesting question then arises: How does 
the sparsity-undersampling ($\rho$~vs.~$\alpha$)
trade-off depend on the choice of dictionary $\vm{D}$? 

The empirical study in 
\cite[Sec.~15 in SI]{Donoho-Maleki-Montanari-2009} 
gave evidence that the worst case $\rho$ vs. $\alpha$
trade-off is quite universal w.r.t different random matrix 
ensembles. 
Analysis in \cite{Kabashima-Wadayama-Tanaka-2009} further revealed
that the typical conditions for perfect $l_{1}$-recovery
are the same for all sensing matrices that are sampled
from the rotationally invariant matrix ensembles.  
Dictionaries with independent identically 
distributed (IID) zero-mean Gaussian elements is one example of this.
But correlations in $\vm{D}$ \emph{can} 
degrade the performance of $l_{1}$-recovery
\cite{Takeda-Kabashima-isit2010}, so 
it is not fully clear how the choice of $\vm{D}$
affects the  $\rho$ vs. $\alpha$ trade-off.

Besides the random / unstructured dictionaries mentioned above, 
the information theoretic approach in \cite{Tulino-etal-isit2011} 
encompasses more general matrix ensembles but does not
consider the $l_1$-reconstruction limit.
Several studies in the literature have also considered the specific 
construction where $\vm{D}$ is formed by concatenating two 
orthogonal matrices \cite{Huo_Theis_1999, 
Donoho_Huo_2001, Elad_Bruckstein_2002, Donoho_Elad_2003, 
Rubinstein_Elad_Dictionary_ProcIEEE_2010}.
Such bi-orthogonal dictionaries are easy to implement and can 
give elegant theoretical insights. 
Unfortunately, the ``mutual coherence'' based 
methods used in these papers provide 
pessimistic, or worst case, thresholds. 
Furthermore, the result are not easy to compare between 
the unstructured and bi-orthogonal
cases.

We consider the analysis of 
the \emph{bi-orthogonal SR setup}
\begin{equation}
\vm{y} = \vm{D} \vm{x} =
\begin{bmatrix}
\vm{O}_{1} & \vm{O}_{2}
\end{bmatrix}
\begin{bmatrix}
\vm{x}_{1} \\
\vm{x}_{2}
\end{bmatrix}
= \vm{O}_{1} \vm{x}_{1}
+ \vm{O}_{2} \vm{x}_{2},
\label{eq:sysmodel_w_A_2}
\end{equation}
where the dictionary is constructed by concatenating two independent matrices 
$\vm{O}_{1}$ and $\vm{O}_{2}$, that are
drawn uniformly according to the Haar measure on the group
of all orthogonal $\M\times\M$ matrices.
We use the \emph{non-rigorous replica method}
(see, e.g., \cite{Kabashima-Wadayama-Tanaka-2009,Tanaka_Raymond_2010,
Guo-Baron-Shamai-2009, Rangan-Fletcher-Goyal-2012} for related works)
to assess $\rho$ for a given $\alpha$, up to which the
$l_{1}$-recovery is successful. This allows a direct comparison 
between the random and bi-orthogonal dictionaries in
average or typical sense. The main result of the paper is the
\emph{sparsity-undersampling trade-off
for the bi-orthogonal SR setup \eqref{eq:sysmodel_w_A_2}}.
We find that this matches the unstructured 
IID Gaussian dictionary when the non-zero components
are uniformly distributed between the two blocks.
Surprisingly, when the non-zero components are concentrated 
more on one block than the other, the \emph{bi-orthogonal dictionaries
can cope with higher overall densities} than the unstructured case.
This extends to the case of general $T$-concatenated 
orthogonal dictionaries as reported elsewhere 
\cite{Lortho-physics}.

\section{Problem Setting}
\label{sec:problem}

Consider the SR problem of finding the sparsest vector
$\vm{x} = [\vm{x}^{\trans}_{1} \; \vm{x}^{\trans}_{2}]^{\trans}\in\R^{\N}$,
given the dense vector $\vm{y}\in\R^{\M}$ and the dictionary
$\vm{D} = [\vm{O}_{1} \; \vm{O}_{2}]\in\R^{\M\times \N}$.
By definition
$\M/\N =  1/2$ and
$\vm{O}_{i}^{\trans} \vm{O}_{i} =\vm{I}_{\M}$ for this setup.
Let $\K_{1}$ and $\K_{2}$ be the number of non-zero elements in $\vm{x}_{1}$ 
and $\vm{x}_{2}$, respectively, so that  $\K = \K_{1} + \K_{2}$ is the 
total number of non-zero elements in $\vm{x}$.  
Denote $\rho = \K / (2\M)$ for the overall sparsity
of the source while $\rho_{1} = \K_{1} / \M$ and 
$\rho_{2} = \K_{2} / \M$ represent the signal densities
of the two blocks.  

It is important to note that $\vm{D}$ in \eqref{eq:sysmodel_w_A_2}
does not belong to the rotationally invariant matrix ensembles
\cite{Kabashima-Wadayama-Tanaka-2009}, and there are complex
dependencies between the elements due to the orthogonality constraints.
The fact that $\vm{O}_{1}^{\trans} \vm{O}_{2} \neq \vm{0}$ 
makes the analysis of the setup highly non-trivial 
(for a sketch, see Appendices~\ref{app:free-energy}~and~\ref{app:mtx-integrals}). 
Thus, only the bi-orthogonal case is considered here and
the analysis of general $T$-concatenated orthogonal 
dictionaries is reported elsewhere 
\cite{Lortho-physics}.

The system is assumed to approach
the large system limit 
$\M,\K_{1},\K_{2}\to\infty$ where the signal densities
$\rho_{1}, \rho_{2}$ 
are finite and fixed.  
We let $\{\vm{x}_{i}\}_{i=1}^{2}$ be independent
sparse random vectors whose components are IID according to
\begin{equation}
\vspace*{-0.15ex}
p_{i}(x) = (1-\rho_{i}) \delta(x) +
\rho_{i} \e^{-x^{2}/2}/\sqrt{2\pi},
\quad i = 1,2.
\label{eq:true_sourcesym_pdf_k}
\vspace*{-0.15ex}
\end{equation}
The convex relaxation of the original problem is considered
and the goal is to find $\vm{x} = [\vm{x}^{\trans}_{1} \; \vm{x}^{\trans}_{2}]^{\trans}$ 
that is the solution to 
\begin{equation}
\min_{\vm{x}_{1}, \vm{x}_{2}}\;
\| \vm{x}_{1} \|_{1} + \| \vm{x}_{2} \|_{1}
\!\quad \text{s.t.} \!\quad \vm{y} = \vm{O}_{1} \vm{x}_{1} + \vm{O}_{2} \vm{x}_{2}.
\label{eq:l1_opt_problem}
\end{equation}
Note that we do not consider the weighted $l_{1}$-reconstruction
analyzed for the rotationally invariant $\vm{D}$ in \cite{Tanaka_Raymond_2010}.  
This corresponds to
the scenario where the user has no prior knowledge about
the relative statistics of the data blocks.
In the next section we find the typical
density $\rho = (\rho_{1} + \rho_{2})/2$ 
for which perfect $l_{1}$-reconstruction is possible under 
the constraint \eqref{eq:sysmodel_w_A_2}.

\section{Analysis}

Let the postulated prior of the sparse vector $\vm{x}_{i}$ be
\begin{equation}
q_{\beta} (\tilde{\vm{x}}_{i})
= \e^{-\beta\|\tilde{\vm{x}}_{i}\|_{1}}, \quad i =1,2,
\label{eq:post_pdf_xi}
\end{equation}
where the components of
$\tilde{\vm{x}}_{i}\in\R^{\M}$ are IID. 
The inverse temperature
$\beta$ is a non-negative parameter. 
Let $q_{\beta} (\tilde{\vm{x}})
= q_{\beta} (\tilde{\vm{x}}_{1})
q_{\beta} (\tilde{\vm{x}}_{2})$
be the postulated prior of $\vm{x}$ in \eqref{eq:sysmodel_w_A_2},
and define a mismatched posterior mean estimator 
\begin{equation}
\label{eq:pme_invtemp_beta}
\left\langle \tilde{\vm{x}} \right \rangle_{\beta}
= 
 Z_{\beta}(\vm{y},\vm{D})^{-1}
\int
\tilde{\vm{x}}\delta(\vm{y} - \vm{D} \tilde{\vm{x}})
q_{\beta} (\tilde{\vm{x}}) 
\mathrm{d} \tilde{\vm{x}}.
\end{equation}
Here
$Z_{\beta}(\vm{y},\vm{D})
= \int\delta(\vm{y} - \vm{D} \tilde{\vm{x}})
q_{\beta} (\tilde{\vm{x}})
\mathrm{d} \tilde{\vm{x}},$
acts as the partition function of the system.
Then, the zero-temperature estimate
$\langle \tilde{\vm{x}} \rangle_{\beta\to\infty}$
is a solution (if at least one exists) to the 
original $l_{1}$-minimization problem  \eqref{eq:l1_opt_problem}. 

Utilizing of one of the standard tools from statistical
physics, namely the non-rigorous \emph{replica method}, 
we study next the behavior of the
estimator \eqref{eq:pme_invtemp_beta}. We accomplish this by
examining the so-called \emph{free energy density $f$} of the system
in the thermodynamic limit $\N\to\infty$.  As a corollary, 
we obtain the critical 
compression threshold for the original 
optimization problem \eqref{eq:l1_opt_problem}
when $\beta\to\infty$.

\subsection{Free Energy}
\label{sec:freeE}

As sketched in Appendix~\ref{app:free-energy},
the free energy density related to
\eqref{eq:pme_invtemp_beta} reads under the replica
symmetric (RS) ansatz
\begin{IEEEeqnarray}{rCl}
f_{\mathsf{rs}} &=&
- \frac{1}{2} \lim_{\beta\to\infty} \frac{1}{\beta} \lim_{\M\to\infty}
\frac{1}{\M} \lim_{\NR\to 0} \frac{\partial}{\partial \NR}
\log \E_{\vm{y},\vm{D}} \{Z^{\NR}_{\beta}(\vm{y},\vm{D})\} \IEEEnonumber\\
  &=&  \frac{1}{2}\,
  \underset{\{\Theta_{1},\Theta_{2}\}}{\mathrm{cextr}}\,
    \sum_{i=1}^{2} T(\Theta_{i}),
  \label{eq:freeE-rmrs-Gauss_main}
\end{IEEEeqnarray}
where
\begin{IEEEeqnarray}{l}
  \label{eq:optfuncT_main}
  T(\Theta_{i}) =
  \frac{\rho_{i} - 2 m_{i} + Q_{i}}{4\chi_{i}} - \frac{Q_{i} \Qhat_{i}}{2}
  + \frac{\chi_{i} \chihat_{i}}{2} + m_{i} \mhat_{i}
  \IEEEeqnarraynumspace\IEEEnonumber\\
   + \!\int\!  (1-\rho_{i})
  \phi (z \sqrt{\chihat_{i}};\,\Qhat_{i})
  + \rho_{i} \phi (z \sqrt{\mhat_{i}^{2}+\chihat_{i}} ;\,\Qhat_{i})
\mathrm{D} z, \IEEEeqnarraynumspace%
\end{IEEEeqnarray}
$\Theta_{i} = \{Q_{i}, \chi_{i}, m_{i}, \Qhat_{i}, \chihat_{i}, \mhat_{i}\}$
is a set of parameters that take values on the extended real line,
$\mathrm{D} z = (2\pi)^{-1/2}\e^{-z^{2}/2} \mathrm{d} z$ is the
Gaussian measure and
\begin{equation}
\phi (h;\,\Qhat) = \min_{x\in\R} \big\{
\Qhat x / 2 - h x + |x| \big\}.
\label{eq:phi-func}
\end{equation}
In contrast to, e.g., 
\cite{Kabashima-Wadayama-Tanaka-2009,Tanaka_Raymond_2010},
here $\mathrm{cextr}_{\Theta} \, g(\Theta)$
is \emph{constrained} extremization over the function $g(\Theta)$
\emph{when $\chi_{1} = \chi_{2},$ needs to be satisfied}.

\begin{remark}
If the dictionary is sampled from 
the rotationally invariant matrix ensembles, the 
RS free energy density reads
\begin{IEEEeqnarray}{l}
f_{\rm rs}\!=\!\frac{1}{2}\!\mathop{\rm extr}_{\{\Theta_1,\Theta_2\}}
\sum_{i=1}^2
\Bigg(\!
\frac{ \rho_i-2m_i+Q_i}{
2\sum_{i=1}^2 \chi_i } 
\!-\frac{Q_i\hat{Q}_i}{2}\!+\frac{\chi_i\hat{\chi}_i}{2}
\!+m_i \hat{m}_i  \IEEEnonumber\\
 + 
\int 
(1-\rho_i)\phi(z\sqrt{\hat{\chi}_i};\hat{Q}_i)
+ \rho_i \phi(z\sqrt{\hat{m}_i^2+\hat{\chi}_i};\hat{Q}_i) 
{\rm D} z
\Bigg),
\end{IEEEeqnarray}
where $\mathrm{extr}$ is an \emph{unconstrained} extremization w.r.t 
$\{\Theta_{1},\Theta_{2}\}$.
\end{remark}

\subsection{Constrained Extremization}

Let us denote $Q(x)= \int_{x}^{\infty}\mathrm{D} z$
for the Q-function
and define 
\begin{equation}
r(h)=
    \sqrt{\frac{h}{2 \pi}} 
    \e^{-\frac{1}{2 h}}
    -(1+h) Q
    \bigg(\frac{1}{\sqrt{h}} \bigg).
\end{equation}
After solving the integrals and 
the optimization problem in \eqref{eq:phi-func},
the function \eqref{eq:optfuncT_main} becomes
\begin{IEEEeqnarray}{rCl}
T(\Theta_{i}) &=&
    \frac{\rho_{i} - 2 m_{i} + Q_{i}}{4\chi_{i}}
      - \frac{Q_{i} \Qhat_{i}}{2}
      + \frac{\chi_{i} \chihat_{i}}{2}
      + m_{i} \mhat_{i}  \IEEEeqnarraynumspace\IEEEnonumber\\
      && +\, \frac{1-\rho_{i}}{\Qhat_{i}}
	r(\chihat_{i})
	+ 
	\frac{\rho_{i}}{\Qhat_{i}}r(\mhat_{i}^{2}+\chihat_{i}).
\IEEEeqnarraynumspace
  \label{eq:optfuncT_lagranges1_main}
\end{IEEEeqnarray}
Introducing the Lagrange multiplier $\eta$ for the constraint
$\chi_{1} = \chi_{2},$
an alternative formulation for the free energy density reads
\begin{IEEEeqnarray}{rCl}
 f_{\mathsf{rs}} &=& %
  \frac{1}{2}\underset{\{\Theta_{1},\Theta_{2}, \eta\}}
  {\mathrm{extr}}\,
  \big\{\eta(\chi_{1}	- \chi_{2})
      +T(\Theta_{1}) + T(\Theta_{2}) \big\},
      \IEEEeqnarraynumspace
\label{eq:FreeEd_wLagrange_main}
\end{IEEEeqnarray}
where the extremization is now an unconstrained problem.
Taking partial derivatives w.r.t all optimization variables and
setting the results to zero yields the identities
\begin{IEEEeqnarray}{lCl}
\Qhat_{i} = \mhat_{i}
\quad &\text{ and }& \quad
\chi_{i} = \frac{1}{2\mhat_{i}}, \quad i = 1,2.
\label{eq:DFm_main}
\end{IEEEeqnarray}
We also find that the expressions
\begin{IEEEeqnarray}{rCl}
  \frac{1}{\mhat_{i}}
  &=&
  \frac{2}{\mhat_{i}} \bigg[
  2(1-\rho_{i})
  Q \bigg(\frac{1}{\sqrt{\chihat_{i}}}\bigg)
  + 2\rho_{i}
  Q \bigg(\frac{1}{
  \sqrt{\smash[b]{\mhat_{i}^{2}}+\chihat_{i}}} \bigg)
  \bigg],	
  \label{eq:DFchik1_main}  \IEEEeqnarraynumspace\\
\chihat_{i} &=& 
\frac{\rho_{i} - 2 m_{i} + Q_{i}}{2\chi_{i}^{2}} 
- \eta \frac{\partial}{\partial \chi_{i}}(\chi_{1}-\chi_{2}),
  \IEEEeqnarraynumspace
  \label{eq:DFchihatk_main}
\end{IEEEeqnarray}
are satisfied by the extremum of
\eqref{eq:FreeEd_wLagrange_main}.
Under perfect reconstruction in  mean square error (MSE) 
sense (see, e.g., \cite{Kabashima-Wadayama-Tanaka-2009,Tanaka_Raymond_2010}
for details), we have $\rho_{i} = Q_{i} = m_{i}$
and $\mhat_{i}\to\infty \implies \chi_{i}\to 0$.
Hence, \eqref{eq:DFchik1_main} simplifies to the condition
\begin{equation}
	2(1-\rho_{i})
  Q \bigg(\frac{1}{\sqrt{\chihat_{i}}}\bigg)
  + \rho_{i}  = \frac{1}{2}.
  \label{eq:DFchik3_main}
\end{equation}
On the other hand, omitting the terms of the order $O(1/\mhat^{3})$, we have
from the partial derivatives of $\Qhat_{i}$ 
and $\mhat_{i}$ 
\begin{align}
 Q_{i}
    &= \rho_{i} 
    - \frac{2\rho_{i}}{\mhat_{i}\sqrt{2 \pi}}
    - 
    \frac{2(1-\rho_{i})}{\mhat_{i}^{2}}
    r(\chihat_{i})
+ \frac{\rho_{i}}{\mhat_{i}^{2}}
    (1+\chihat_{i}),
  \label{eq:DFQk1_main} \\ 
 m_{i}
    &=
    \rho_{i} - \frac{\rho_{i}}{\mhat_{i}\sqrt{2 \pi}},
    \IEEEeqnarraynumspace
  \label{eq:DFmk1_main}
  \end{align}
respectively, where we used \eqref{eq:DFm_main} to simplify the expressions.
Plugging the above to 
\eqref{eq:DFchihatk_main}
and using again \eqref{eq:DFm_main} yields
\begin{IEEEeqnarray}{l}
\chihat_{i} = (-1)^{i} \eta 
+ 2\rho_{i} (1+\chihat_{i}) -  4(1-\rho_{i}) r(\chihat_{i}).
    \IEEEeqnarraynumspace
    \label{eq:DFchihatk_main2}
\end{IEEEeqnarray}

Before stating the final result, 
let us introduce a real parameter $\mu \in [0,1]$ and 
assume without loss of generality that $\rho_{1} = \mu \rho_{2}$.
Then the per-block densities can be written as
\begin{equation}
\rho_{1} = \frac{2 \mu}{1+\mu}\rho
\quad \text{and} \quad
\rho_{2} = \frac{2}{1+\mu}\rho,
\end{equation}
where $\rho = \rho(\mu)$ is the overall density of the source.
The parameter $\mu$ determines thus
how uniformly the non-zero components are distributed between 
the two blocks: $\mu=1$ means fully uniformly, 
$\mu=0$ implies that all non-zero
components are in the second block.

\begin{mainresult}
Let $\vm{x}\in\R^{2 \M}$,
$\vm{D} \in\R^{\M \times 2 \M}$ and
$\vm{y} = \vm{D} \vm{x}$ as in \eqref{eq:sysmodel_w_A_2}.
Given the parameter $\mu\in [0,1]$, the typical 
density $\rho(\mu)$ of the solution to
the optimization problem
\begin{equation*}
\underset{\vm{x}=[\vm{x}_{1} \; \vm{x}_{2}]^{\trans}\in\R^{2\M}}{\arg \min} \;
\| \vm{x}_{1} \|_{1} + \| \vm{x}_{2} \|_{1}
\quad \text{s.t.} \quad \vm{y} = \vm{D} \vm{x},
\end{equation*}
is determined in the large system limit
by the solutions of the following set of coupled equations
\begin{IEEEeqnarray}{rCl}
\chihat_{1} &=&
\bigg[Q^{-1}\bigg( 
\frac{1}{4} - 
\frac{2 \mu \rho}{1+\mu}
\bigg[
\frac{1}{2} - Q\bigg(\frac{1}{\sqrt{\chihat_{1}}} \bigg)
\bigg] \bigg)\bigg]^{-2},  \IEEEeqnarraynumspace\\
\eta &=& 
\frac{4 \mu \rho}{1+\mu}
\big[1+\chihat_{1}
+ 2 r(\chihat_{1})
\big] 
- 4 r(\chihat_{1})
-\chihat_{1},
\IEEEeqnarraynumspace \\
\chihat_{2} &=& 
\frac{4 \rho}{1+\mu}
\big[1+\chihat_{2} + 2 r(\chihat_{2})\big]
 - 4 r(\chihat_{2}) + \eta, 
\IEEEeqnarraynumspace \\
\rho
   &=& (1+\mu)\bigg[\frac{1}{2}- 2 Q \bigg(\frac{1}{\sqrt{\chihat_{2}}}\bigg)\bigg]
   \bigg/
   	\bigg[2-4
       Q \bigg(\frac{1}{\sqrt{\chihat_{2}}}\bigg)\bigg],
       \IEEEeqnarraynumspace
\end{IEEEeqnarray}
where $Q^{-1}$ is the functional inverse of the Q-function.
For uniform sparsity, that is, $\mu = 1$ and $\rho_{1} = \rho_{2}$, 
we have $\eta = 0$,
$\chihat_{1} = \chihat_{2}$ and $\chi_{1} = \chi_{2}$ always.
The critical density is thus the same as for the dictionary that is drawn 
from the ensemble of rotationally invariant matrices.
\end{mainresult}

\subsection{Numerical Examples}

Given the dictionary $\vm{D}$ is drawn 
from the ensemble of rotationally invariant matrices,
the critical density for $l_{1}$-recovery is 
known to be independent of the block densities 
$\{\rho_{1},\rho_{2}\}$ and given by
$\rho = 0.19284483309074016\!\ldots$
for all $\mu \in [0,1]$.  For the bi-orthogonal $\vm{D}$, 
the threshold is the same only for the case of uniform sparsity 
$\mu = 1$.  For general $\mu$ we obtain  different thresholds,
as plotted in Fig.~\ref{fig:plot1}.  Note that $\rho(\mu)$ is 
a decreasing function of $\mu$, implying that the more 
concentrated the non-zero components are in one block, the bigger the 
benefit of using the bi-orthogonal dictionary.
We also carried out
numerical simulations for the IID Gaussian and
bi-orthogonal $\vm{D}$ using 'linprog' from Matlab Optimization
Toolbox.  The results are plotted in Fig.~\ref{fig:plot2},
where for each value of $N = 16, 18, \ldots, 50$, there are $10^{6}$
realizations of the SR problem.
Cubic curves are fitted to the data using
nonlinear least-squares regression. 
The critical density for the bi-orthogonal case is 
predicted by the replica method to be
$\rho(0) = 0.22666551758496698\!\ldots$ and we observe 
that the simulations match the analysis up to the third decimal
place.

\begin{figure}[t]
\centering
\includegraphics[width=0.9\columnwidth]{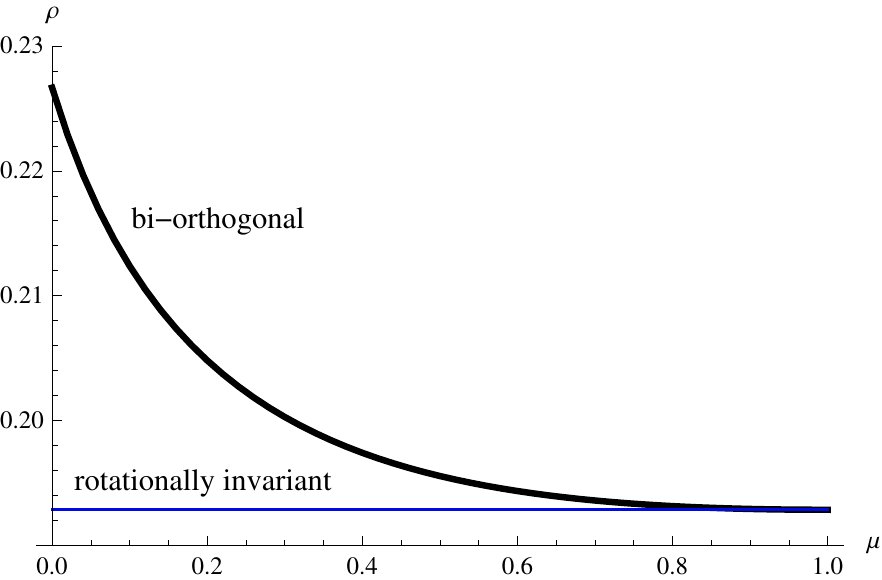}
\vspace*{-1.5ex}
\caption{Critical density for bi-orthogonal and rotationally 
invariant $\vm{D}$.  The parameter $\mu\in[0,1]$ determines
how uniformly the non-zero components are distributed between 
the two blocks ($\mu=1$ fully uniform, $\mu=0$ all non-zero
components are in the second block).  The user has no knowledge 
about $\mu$.}
\label{fig:plot1}
\end{figure}

\begin{figure}[t]
\centering
\includegraphics[width=0.95\columnwidth]{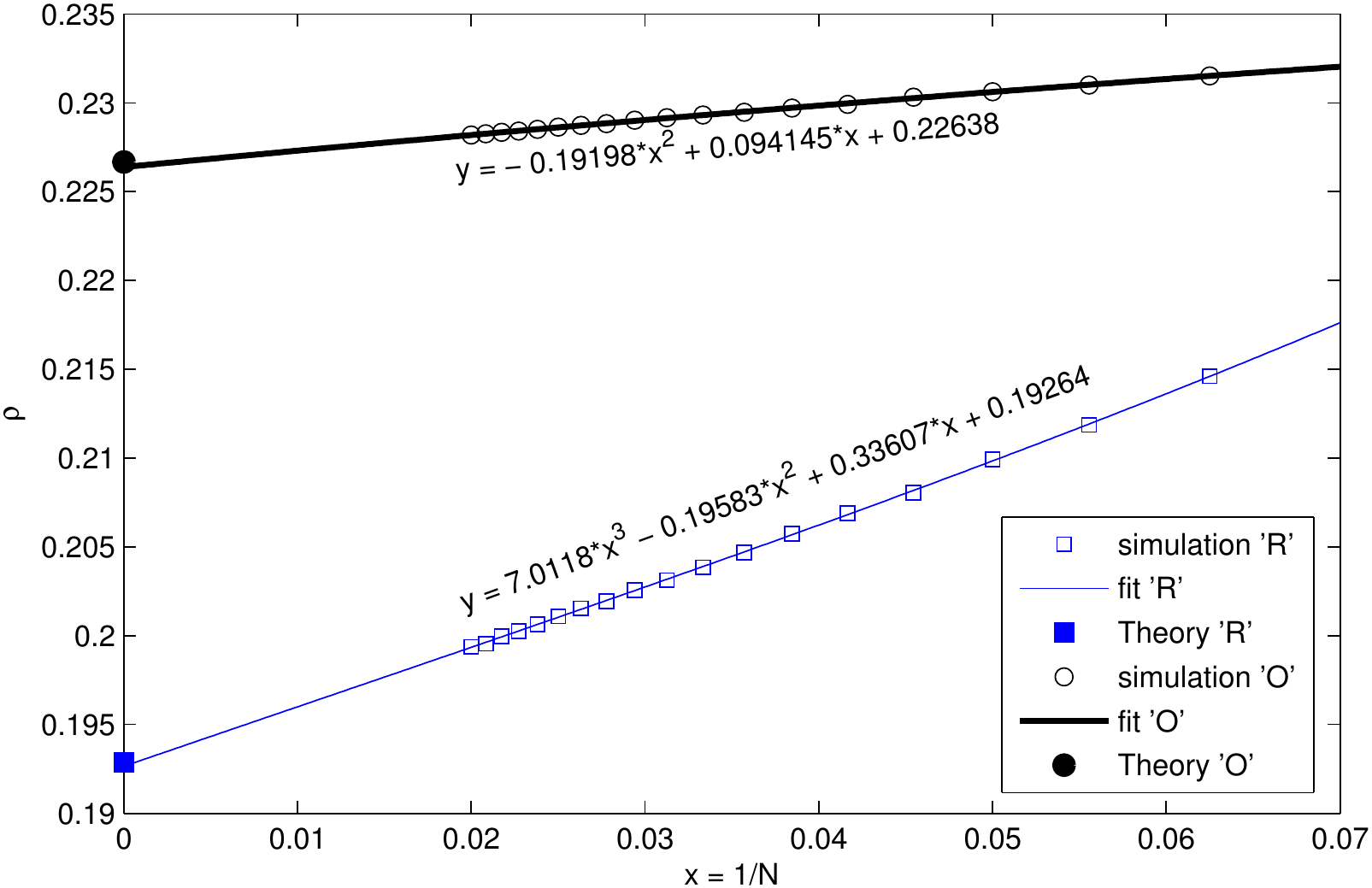}
\vspace*{-1.5ex}
\caption{Critical density given $\mu = 0$, that is, $\rho_{1} = 0, \rho_{2} = 2 \rho$
for finite sized systems.  Here 'R' means
rotationally invariant $\vm{D}$ and 'O' the bi-orthogonal case.  Each point
is averaged over $10^{6}$ realizations of the optimization problem.
The filled markers at $x = 0$ are the predictions given by the replica analysis.}
\label{fig:plot2}
\end{figure}

\section{Conclusions and Discussion}

The  sparsity-undersampling trade-off
for the bi-orthogonal SR setup \eqref{eq:sysmodel_w_A_2} was 
studied.  For uniformly distributed non-zero components,
there is no difference in compression ratio if we replace
the rotationally invariant dictionary $\vm{D}\in\R^{\M \times 2\M}$
by a concatenated matrix
$\vm{D} = [\vm{O}_{1} \;
\vm{O}_{2}]\in\R^{\M\times 2\M}$, where
$\vm{O}_{1}, \vm{O}_{2}$ are independent and drawn
uniformly according to the Haar measure on the  group
of all orthogonal $\M\times\M$ matrices.  
For non-uniform block sparsities, however, the bi-orthogonal
dictionaries were found to be beneficial compared to the 
unstructured random dictionaries.

\appendices

\section{Free Energy}
\label{app:free-energy}

Following \cite{Kabashima-Wadayama-Tanaka-2009,
Tanaka_Raymond_2010}, we use the replica trick and write the 
free energy density as
\begin{equation}
f = - \frac{1}{2}\lim_{\beta\to\infty}
\frac{1}{\beta} \lim_{\NR\to 0} \frac{\partial}{\partial \NR}
\lim_{\M\to\infty} \frac{1}{\M} \log  \Xi^{(\NR)}_{\beta,\M},
\label{eq:FreeEd3}
\end{equation}
where denoting 
$\Delta\vm{x}_{i}^{\nrbra} = \vm{x}^{[0]}_{i} - \vm{x}^{\nrbra}_{i}, \nr = 0,1,\ldots,\NR$,
\begin{equation}
\Xi^{(\NR)}_{\beta,\M} 
=
\E  \lim_{\tau\to 0^{+}}  \!
\frac{1}{\tau^{\frac{\NR\M}{2}}}
\E
\Bigg\{
\e^{-\frac{1}{2\tau} \sum_{\nr=1}^{\NR} \|
\vm{O}_{1} \Delta\vm{x}_{1}^{\nrbra}
+ \vm{O}_{2} \Delta\vm{x}_{2}^{\nrbra}\|^{2}} 
\, \Bigg| \, \mathcal{X} \Bigg\}.
\label{eq:Axi1_3}
\end{equation}
For $i = 1,2$, the vectors
$\{\vm{x}_{i}^{\nrbra}\}_{\nr=1}^{\NR}$ 
 are IID conditioned on $\vm{D}$ and have the same density
 \eqref{eq:post_pdf_xi} as $\tilde{\vm{x}}_{i}$.
Furthermore, the elements of the vectors $\vm{x}^{[0]}_{1}$ and
$\vm{x}^{[0]}_{2}$ are independently drawn according to
  $p_{1}$ and $p_{2}$ as given in
  \eqref{eq:true_sourcesym_pdf_k}, and $\mathcal{X} =
\{\vm{x}^{\nrbra}_{1},\vm{x}^{\nrbra}_{2}\}_{\nr=0}^{\NR}$.

Let us concentrate on $\Xi^{(\NR)}_{\beta,\M}$ and the inner expectation in
\eqref{eq:Axi1_3}, which is over the orthogonal matrices
$\vm{O}_{1}$ and $\vm{O}_{2}$ given $\mathcal{X}$.
Since $\vm{O}_{i}$ are orthogonal, the average 
affects only the cross-terms of the form
$(\vm{u}^{\nrbra}_{1})^{\trans}\vm{u}^{\nrbra}_{2}$
where $\vm{u}^{\nrbra}_{i} = \vm{O}_{i}\Delta\vm{x}_{i}^{\nrbra}$.
Define matrices $\vm{S}_{i} \in \R^{\NR \times \NR}$ for $i = 1,2,$
whose $(\nr,\nrb)$th element
\begin{IEEEeqnarray}{rCl}
S_{i}^{[\nr,\nrb]} 
&=& Q_{i}^{[0,0]}
- Q_{i}^{[0,\nrb]}
- Q_{i}^{[\nr,0]}
+ Q_{i}^{[\nr,\nrb]}, \quad i= 1,2 
\IEEEeqnarraynumspace
\label{eq:symbolSab}
\end{IEEEeqnarray}
is the empirical covariance between the elements of 
$\Delta\vm{x}_{i}^{\nrbra}$ and 
$\Delta\vm{x}_{i}^{\nrbbra}$, written in terms of 
the empirical covariances
\begin{equation}
\label{eq:symbolQab}
Q_{i}^{[\nr,\nrb]} =
\M^{-1}(\vm{x}_{i}^{\nrbra})^{\trans}\vm{x}_{i}^{\nrbbra},
\quad \nr,\nrb=0,1,\ldots,\NR.
\end{equation}
between the components of the $\nr$th and
$\nrb$th replicas of $\vm{x}_{i}$.
For analytical tractability, we make 
the standard replica symmetry (RS)
assumption on the correlations \eqref{eq:symbolQab}, i.e.,
$r_{i} = Q_{i}^{[0,0]}$, 
$m_{i} = Q_{i}^{[0,\nrb]} = Q_{i}^{[\nr,0]} \, \forall a,b\geq 1$,
$Q_{i} = Q_{i}^{[\nr,\nr]}\, \forall a\geq 1$ and
$q_{i} = Q_{i}^{[\nr,\nrb]}\, \forall a \neq b\geq 1$.
The RS free energy density is denoted $f_{\mathsf{rs}}$ and
we remark that it does not match
$f$ if the system is replica symmetry breaking.
Under the RS assumption,
\begin{equation}
\label{eq:Smtx_i_RS}
\vm{S}_{i} 
= S_{i}^{[1,2]}\vm{1}_{\NR}\vm{1}_{\NR}^{\trans}
  + (S_{i}^{[1,1]}-S_{i}^{[1,2]})\vm{I}_{\NR},
  \qquad i=1,2,
  \end{equation}
where $\vm{1}_{\NR}\in\R^{\NR}$ is the vector of all-ones, and
we may write the inner expectation in \eqref{eq:Axi1_3} as
\begin{equation}
\e^{-\frac{\NR\M}{2\tau}( S_{1}^{[1,1]} + S_{2}^{[1,1]})}
\E \Big\{ 
\e^{- \frac{1}{\tau}
\sum_{\nr=1}^{\NR}
(\vm{u}^{\nrbra}_{1})^{\trans}\vm{u}^{\nrbra}_{2}} 
\, \Big| \, \mathcal{X} \Big\}. \IEEEeqnarraynumspace
\label{eq:avg_over_U1U2_1}
\end{equation}
Using Lemma~\ref{lemma:mtx_int_2} 
and taking the limit $\tau \to 0^{+}$ leads to
\begin{IEEEeqnarray}{l}
\label{eq:Axi1_4}
\Xi^{(\NR)}_{\beta,\M} = \int \e^{-\M
G^{(\NR)}}\prod_{\NR=1}^{\NR}
\e^{-\beta(\|\vm{x}^{\nrbra}_{1}\|_{1}
+\|\vm{x}^{\nrbra}_{2}\|_{1})}
\mathrm{d} \vm{x}^{\nrbra}_{1} \mathrm{d} \vm{x}^{\nrbra}_{2},
\IEEEeqnarraynumspace
\end{IEEEeqnarray}
where $G^{(\NR)} = \lim_{\tau\to 0^{+}}G^{(\NR)}_{\tau}$.  The 
function $G^{(\NR)}_{\tau}$ given in \eqref{eq:funcGu_simplified}
at the top of the next page
is implicitly a function of both $\vm{S}_{1}$ and $\vm{S}_{2}$.
To obtain \eqref{eq:funcGu_simplified}
we first used \eqref{eq:lemma_2}, then applied 
\eqref{eq:approx_IM}. 
Finally, some
algebraic manipulations give the reported result.
\begin{figure*}[!t]
\normalsize
\begin{IEEEeqnarray}{rCl}
G^{(\NR)}_{\tau}
&=& \frac{1}{2\tau}
  \Big(\sqrt{S_{1}^{[1,1]}-S_{1}^{[1,2]}+ \NR S_{1}^{[1,2]}}
  - \sqrt{S_{2}^{[1,1]}-S_{2}^{[1,2]}+ \NR S_{2}^{[1,2]}} \Big)^{2}
   + \, \frac{\NR-1}{2\tau}
  \Big(\sqrt{S_{1}^{[1,1]}-S_{1}^{[1,2]}}
  - \sqrt{S_{2}^{[1,1]}-S_{2}^{[1,2]}} \Big)^{2}	\IEEEeqnarraynumspace
\IEEEnonumber\\
  && + \, \frac{1}{4} \log\Big[
  \big(S_{1}^{[1,1]}-S_{1}^{[1,2]}+ \NR S_{1}^{[1,2]}\big)
  \big(S_{2}^{[1,1]}-S_{2}^{[1,2]}+ \NR S_{2}^{[1,2]}\big)
  \Big]	
   + \, \frac{\NR-1}{4} \log\Big[
  \big(S_{1}^{[1,1]}-S_{1}^{[1,2]}\big)
  \big(S_{2}^{[1,1]}-S_{2}^{[1,2]}\big)
  \Big],
  \IEEEeqnarraynumspace 
  \label{eq:funcGu_simplified}
\vspace*{0em}
\end{IEEEeqnarray}
\hrulefill
\vspace*{-0.5em}
\end{figure*}

The problem with the limit 
$G^{(\NR)} = \lim_{\tau\to 0^{+}}G^{(\NR)}_{\tau}$ is 
that it diverges and the free energy density grows
without bound which is an undesired result.
To keep $G^{(\NR)}$ and the free energy density
finite as $\tau\to 0^{+}$, we
pose the constraints 
\begin{IEEEeqnarray}{rCl}
  \label{eq:constr_1}
  S_{1}^{[1,1]}-S_{1}^{[1,2]}+ \NR S_{1}^{[1,2]} &=&
  S_{2}^{[1,1]}-S_{2}^{[1,2]}+ \NR S_{2}^{[1,2]},
  \IEEEeqnarraynumspace\\
  S_{1}^{[1,1]}-S_{1}^{[1,2]} &=&
  S_{2}^{[1,1]}-S_{2}^{[1,2]},
  \IEEEeqnarraynumspace
  \label{eq:constr_2}
\end{IEEEeqnarray}
on the elements of the replica symmetric matrices 
$\vm{S}_{1}, \vm{S}_{2}$.
Given \eqref{eq:constr_1} and 
\eqref{eq:constr_2} are satisfied, we get 
in the limit $\tau \to 0^{+}$ 
the expression for 
$G^{(\NR)} = G^{(\NR)}_{1} + G^{(\NR)}_{2}$ in terms of
\begin{IEEEeqnarray}{rCl}
  \label{eq:funcGuK_1}
  G^{(\NR)}_{i} &=&
  \frac{1}{4} \log
  \big(Q_{i}	- q_{i} + \NR (r_{i} - 2 m_{i} + q_{i}) \big) 
   \IEEEeqnarraynumspace \IEEEnonumber\\
  && \; + \frac{\NR-1}{4} \log
  (Q_{i} - q_{i}), \qquad i = 1,2. 	 \IEEEeqnarraynumspace
\end{IEEEeqnarray}
Comparing \eqref{eq:funcGuK_1} to 
\cite[eq.~(A.4)]{Kabashima-Wadayama-Tanaka-2009} reveals  
that the corresponding terms for rotationally invariant and bi-orthogonal 
$\vm{D}$ match up to vanishing constants.
Furthermore, in the limit $\NR\to 0$ the equalities
\eqref{eq:constr_1} and \eqref{eq:constr_2} 
are equivalent to the condition $\chi_{1} = \chi_{2},$ where we denoted
$\chi_{i} = \beta (Q_{i}-q_{i})$ for notational convenience.
This provides the relevant 
constraint for the evaluation of the RS free energy, as stated 
in Section~\ref{sec:freeE}.  

The next task would be to average \eqref{eq:Axi1_4} 
over the correlations \eqref{eq:symbolQab} using the 
theory of large 
deviations and saddle-point integration.
But since the effect of the
bi-orthogonal sensing matrix $\vm{D}$ has been reduced to the 
above constraint,
we omit the calculations here due to space constraints. 
For details, see \cite[Appendix~A]{Kabashima-Wadayama-Tanaka-2009}
and \cite{Lortho-physics}.

\section{Matrix Integration}
\label{app:mtx-integrals}

\begin{lemma}
  \label{lemma:mtx_int_1}
Let $\vm{O}_{1}$ and $\vm{O}_{2}$ be independent and drawn
uniformly according to the Haar measure
on the group of all orthogonal $\M\times\M$ matrices
as in  \eqref{eq:sysmodel_w_A_2}.
Given vectors $\vm{x}_{1},\vm{x}_{2}\in\R^{\M},$ 
denote $\|\vm{x}_{i}\|^{2} = \M r_{i}$, for $i = 1,2.$
Then 
\begin{equation}
I_{\M} (r_{1},r_{2};c)
=\E_{\vm{O}_{1},\vm{O}_{2}}
\e^{c \vm{x}_{1}^{\trans}\vm{O}_{1}^{\trans}\vm{x}_{2}\vm{O}_{2}}
=\E_{\vm{u}_{1}, \vm{u}_{2}}\e^{c\vm{u}_{1}^{\trans}\vm{u}_{2}},
\label{eq:int_largeN_U1U2_1}
\end{equation}
where $c\in R$ and vectors $\vm{u}_{1},\vm{u}_{2}\in\R^{\M}$ are 
independent and  uniformly distributed on the 
hyper-spheres at the boundaries of $\M$ dimensional balls
with radiuses $R_{1}=\sqrt{\M r_{1}}$ and
$R_{2}=\sqrt{\M r_{2}}$, respectively.
Furthermore, 
\begin{IEEEeqnarray}{l}
\label{eq:Fxy}
F (r_{1},r_{2};c) =
\lim_{\M\to\infty} \M^{-1}\log I_{\M} (r_{1},r_{2};c) \IEEEnonumber\\
\; = \frac{\sqrt{1+4 c^{2} r_{1} r_{2}}}{2}
- \frac{1}{2} \log \left( \frac{1+ \sqrt{1+4 c^{2} r_{1} r_{2}}}{2} \right)
- \frac{1}{2}
\label{eq:IM_exact}
\IEEEeqnarraynumspace \\ 
\;\approx \sqrt{c^{2}r_{1} r_{2}} - \log (c^{2} r_{1} r_{2})/4,
\quad \text{for $c^{2}r_{1} r_{2} \gg 1$}. \IEEEeqnarraynumspace
\label{eq:approx_IM}
\end{IEEEeqnarray}

\end{lemma}

\begin{proof}
  Let $\vm{u}_{i} = \vm{O}_{i}\vm{x}_{i}$ where $\{\vm{x}_{i}\}_{i=1}^{2}$
  are fixed and
  $\{\vm{O}_{i}\}_{i=1}^{2}$ independent and drawn
  uniformly according to the Haar measure
  on the group of all orthogonal $\M\times\M$ matrices.  
  Since $\|\vm{u}_{i}\|^{2} = \M r_{i}$
  and $\vm{O}_{i}$ rotate the vectors $\vm{u}_{i}$
  uniformly in all directions, $\vm{u}_{i}$ is uniformly distributed 
  on the hyper-sphere at the boundaries of an $\M$ dimensional ball
  having radius $R_{i}=\sqrt{\M r_{i}}$, providing the second equality in
  \eqref{eq:int_largeN_U1U2_1}.
  
  To assess the second part of the lemma, the joint measure of 
  $(\vm{u}_{1},\vm{u}_{2})$ reads
  $p(\vm{u}_{1}; r_{1})p(\vm{u}_{2}; r_{2})
  \mathrm{d}\vm{u}_{1}\mathrm{d}\vm{u}_{2}$,
  where
    \begin{equation}
    \label{eq:pdf_u}
    p(\vm{u}; r) = Z(r)^{-1} \delta(\|\vm{u}\|^{2}-\M).
    \IEEEeqnarraynumspace
    \end{equation}
  The normalization constant
  $Z(r)$ in  \eqref{eq:pdf_u} is the volume of the hypersphere 
  in which $\vm{u}$ is constrained to.  Using 
  Stirling's formula
  for large $\M$, we get up to a vanishing term $O (1/\M)$ 
  \begin{IEEEeqnarray}{rCl}
    Z(r)
    = (2 \pi \e r)^{\M/2} / \sqrt{\pi r}.
    \IEEEeqnarraynumspace
  \end{IEEEeqnarray}
  With the help of Laplace transform, we write
  \begin{equation}
  \delta(x-a)
  = \frac{1}{4 \pi \im} \int_{\gamma-\im\infty}^{\gamma+\im\infty}
  \e^{-\frac{1}{2}s(x-a)} \mathrm{d} s, \qquad \gamma\in\R,
  \label{eq:delta_L_transform_1}
  \end{equation}
  so that using \eqref{eq:pdf_u}~--~\eqref{eq:delta_L_transform_1},
  the latter expectation in \eqref{eq:int_largeN_U1U2_1} 
  becomes
  \begin{IEEEeqnarray}{l}
    \frac{(4 \pi \im)^{-2}}{Z(r_{1}, r_{2})}
    \int \e^{
    c \vm{u}_{1}^{\trans}\vm{u}_{2}
    -\sum_{i=1}^{2}(\|\vm{u}_{i}\|^{2}-\M r_{i})s_{i}/2}
	\prod_{i=1}^{2}
    \mathrm{d} \vm{u}_{i} 
    \mathrm{d} s_{i}
    \IEEEeqnarraynumspace\IEEEnonumber\\ [-0.5ex]
    \quad =
    \frac{(4\im)^{-2}
    \sqrt{r_{1} r_{2}}
    }{\pi\e^{\M}
        (r_{1} r_{2})^{\M/2}}
    \int
  	\frac{\e^{\M \frac{s_{1}r_{1}+s_{2}r_{2}}{2}}}{(s_{1} s_{2} - c^{2})^{\M/2}}
	\mathrm{d} s_{1} \mathrm{d} s_{2}, \IEEEeqnarraynumspace
    \label{eq:int_largeN_U1U2_3}
  \end{IEEEeqnarray}
  where we used Gaussian integration to obtain \eqref{eq:int_largeN_U1U2_3}.
  Since $\M\to\infty$, we next apply saddle-point integration to solve
  the integrals w.r.t $s_{1}$ and $s_{2}$.  After canceling the vanishing 
  terms, 
  \begin{IEEEeqnarray}{l}
  \label{eq:IM_x1x2_1}
  \lim_{\M\to\infty} \M^{-1}\log I_{\M} (r_{1},r_{2}; c)
  \IEEEnonumber\\ [-1ex]
  = -1 
  - \frac{1}{2}\sum_{i=1}^{2} \log r_{i} 
  + \frac{1}{2} \, \underset{s_{1},s_{2}}{\mathrm{extr}}\:
  \bigg\{
  \sum_{i=1}^{2}s_{i}r_{i}
  - \log(s_{1} s_{2} - c^{2})
  \bigg\}, \IEEEnonumber \\ [-0.5ex]
\vspace*{-1ex}
  \end{IEEEeqnarray}
  and \eqref{eq:IM_exact} follows by solving 
  the extremization, and \eqref{eq:approx_IM}  by 
  neglecting the terms that are of the order unity.
\end{proof}

\begin{lemma}
  \label{lemma:mtx_int_2}
Let $\{\vm{O}_{i}\}_{i=1}^{2}$ be as in Lemma~\ref{lemma:mtx_int_1}, and
$\Delta\vm{x}_{i}^{\nrbra}$
for $i=1,2$ and $\nr =1,\ldots,\NR$ 
as in \eqref{eq:Axi1_3}.  Then, under RS ansatz
\begin{IEEEeqnarray}{l}
\lim_{\M\to\infty} M^{-1}
\log \,
\E_{\vm{O}_{1},\vm{O}_{2}} \big\{
\e^{c\sum_{\nr=1}^{\NR}
(\vm{O}_{1} \Delta\vm{x}_{1}^{\nrbra})^{\trans}
(\vm{O}_{2} \Delta\vm{x}_{2}^{\nrbra})} \, \big| \, \mathcal{X}
\big\} \IEEEeqnarraynumspace \IEEEnonumber \\ [1ex]
 =
  F\big(
  S_{1}^{[1,1]}-S_{1}^{[1,2]}+ \NR S_{1}^{[1,2]},
  S_{2}^{[1,1]}-S_{2}^{[1,2]}+ \NR S_{2}^{[1,2]}; c\big)
  \IEEEeqnarraynumspace \IEEEnonumber\\
   \qquad + (\NR-1)
  F\big(
  S_{1}^{[1,1]}-S_{1}^{[1,2]},
  S_{2}^{[1,1]}-S_{2}^{[1,2]}; c\big),
  \IEEEeqnarraynumspace
  \label{eq:lemma_2}
\end{IEEEeqnarray}
where $c\in R$ and $F(r_{1}, r_{2}; c)$ is given in \eqref{eq:Fxy}.
\end{lemma}

\begin{proof}
Denote
$\vm{u}^{\nrbra}_{i} = \vm{O}_{i} \Delta\vm{x}_{i}^{\nrbra}$
for all $i = 1,2$ and $\nr = 1,\ldots,\NR$.
Given $\mathcal{X}$, $\vm{u}^{\nrbra}_{i}$ lie on the surfaces of 
hyper-spheres as in the proof of Lemma~\ref{lemma:mtx_int_1}.
The RS ansatz guarantees that $\vm{u}_i^{[a]}$ can be 
expressed as $[\vm{u}_i^{[1]} \ \vm{u}_i^{[2]} \ \cdots \ \vm{u}_i^{[u]}]
=[\tilde{\vm{u}}_i^{[1]} \ \tilde{\vm{u}}_i^{[2]} \ \cdots \ \tilde{\vm{u}}_i^{[u]}]
\vm{E}^{\trans}$, where $\{\tilde{\vm{u}}_i^{[a]}\}$ is a set of  vectors
that satisfies $M^{-1} \tilde{\vm{u}}_i^{[a]} \cdot \tilde{\vm{u}}_i^{[b]}=0$ if $a\ne b$ and 
\begin{equation} 
\frac{1}{\M} \tilde{\vm{u}}_i^{[a]} \cdot \tilde{\vm{u}}_i^{[b]}
=\left \{
\begin{array}{ll}
uS_i^{[1,2]}+(S_i^{[1,1]}-S_i^{[1,2]} ) & \mbox{if $a=b=1$}; \cr
S_i^{[1,1]}-S_i^{[1,2]} &  \mbox{if $a=b\ge 2$}. 
\end{array}
\right . 
\label{eq:app-vecW-corrs}
\end{equation}
The matrix $\vm{E}=[u^{-1/2} \vm{1}_u \ \vm{e}_2 \ \cdots \ \vm{e}_u ]$ 
provides an orthonormal 
basis that is independent of index $i$.  
This indicates that the expectation in (45) can be assessed w.r.t. $\{\tilde{\vm{u}}_i^{[a]}\}$ 
instead of the original non-orthogonal set $\{\vm{u}_i^{[a]}\}$.  
The orthogonality allows us to independently evaluate the expectation for 
each replica index $a$ when $u \ll M$. 
Using Lemma 1 and \eqref{eq:app-vecW-corrs} completes the proof. 
\end{proof}

\section*{Acknowledgment}
This work was supported in part by the 
Swedish Research Council under VR Grant 621-2011-1024
and by grants from
the JSPS (KAKENHI Nos. 22300003 and 22300098).

\bibliography{./biblio_saikat_CS}
\bibliographystyle{IEEEtran}

\end{document}